\documentclass[a4paper]{easychair}
\usepackage{bussproofs}
\usepackage{amsmath,amssymb,amsthm}
\usepackage{mathrsfs}
\usepackage{stmaryrd}

\newtheorem{theorem}{Theorem}[section]
\newtheorem{lemma}[theorem]{Lemma}

\newtheorem{corollary}[theorem]{Corollary}

\theoremstyle{plain}
\newtheorem{definition}[theorem]{Definition}
\newtheorem{example}[theorem]{Example}
\newcommand{\mf}{\mathfrak}
\newcommand{\msf}{\mathsf}

\newcommand{\mbb}{\mathbb}
\newcommand{\mrm}{\mathrm}
\newcommand{\bu}{\bullet}
\newcommand{\mc}{\mathcal}

\newcommand{\imp}{\rightarrow}

\newcommand{\seq}{\Rightarrow}

\newcommand{\bb}{\blacksquare}
\newcommand{\sub}{\subseteq}

\newcommand{\w}{\widehat}
\newcommand{\ve}{\varnothing}
\newcommand{\D}{\Diamond}
\newcommand{\B}{\Box}

\newcommand{\tup}[1]{\langle #1 \rangle}

\newcommand{\den}[1]{\llbracket{#1}\rrbracket}

\usepackage{tikz}
\usetikzlibrary{patterns,arrows}

\title{\Large The Finite Model Property of Quasi-transitive Modal Logic}

\author{
Zhe Lin
\and
Minghui Ma
}

\institute{
Institute of Logic and Cognition, Sun Yat-sen University, Guangzhou, China\\
\email{\{linzhe8,mamh6\}@mail.sysu.edu.cn}
}

\authorrunning{Z. Lin and M. Ma}

\titlerunning{The Finite Model Property of Quasi-transitive Modal Logic}

\begin{document}
\maketitle
\bigskip
\begin{abstract}
The finite model property of quasi-transitive modal logic $\mathsf{K}_2^3=\mathsf{K}\oplus \Box\Box p\rightarrow \Box\Box\Box p$ is established. This modal logic is conservatively extended to the tense logic $\mathsf{Kt}_2^3$. We present a Gentzen sequent calculus $\msf{G}$ for $\mathsf{Kt}_2^3$. The sequent calculus $\msf{G}$ has the finite algebra property by a finite syntactic construction. It follows that $\mathsf{Kt}_2^3$ and $\mathsf{K}_2^3$ have the finite model property.
\end{abstract}

\section{Introduction}
\label{sect:introduction}
Modal reduction principles (MRPs) are modal formulas of the form $\msf{M} p\imp \msf{N}p$ where $\msf{M}, \msf{N}$ are finite (possibly empty) sequences of modal operators $\Box$ or $\D$. Fitch \cite{Fitch} investigated the problem of correspondence between MRPs and first-order properties. Van Benthem \cite{vB1976} proved that every MRP corresponds to a first-order relational property over the class of all transitive frames. Besides the correspondence theory, the {\em finite model property} (FMP) of normal modal logics generated by MRPs is also much concerned in the literature. 

For any normal modal logic $\Lambda$ and a set $\Sigma$ of modal formulas, let $\msf{NExt}(\Lambda)$ be the class of all normal modal logics extending $\Lambda$, and let $\Lambda\oplus\Sigma$ be the normal extension of $\Lambda$ by adding all formulas in $\Sigma$ as axioms. Using the method of canonical formulas, Zakharyashev \cite{Zak1997} proved that all logics in $\msf{NExt}(\msf{K4})$ axiomatized by MRPs have the FMP. However, 
the FMP of normal modal logics axiomatizable by MRPs over the least normal modal logic $\msf{K}$ is a longstanding open problem (cf.~\cite[p.452]{FZ2007}). 
In particular, it is unknown whether all normal modal logic of the form $\msf{K}_n^m =\msf{K}\oplus\Box^n p\imp \Box^m p$ ($n\neq m\geq 1$)
have the FMP. This most intriguing open problem in modal logic was highlighted by Zakharyashev \cite{Zak1997} as follows:
\begin{quote}
{\em Unfortunately, the technical apparatus developed is applicable only to logics with transitive frames, and the situation of extensions of $\msf{K}$ by modal reduction principles, even by axioms $\Box^n p \imp \Box^mp$ still remains unclear. I think at present this is one of the major challenges in completeness theory.}
\end{quote}
This problem has a long history and was traced back to Krister Segerberg in 1970s by Chagrov and Zakharyashev \cite[11.8 Notes]{CZ1997}. It is worth mentioning that Gabbay \cite{Gabbay1972} used a general filtration method to show the FMP of modal logics $\msf{K}\oplus \B p\imp \B^m p$ where $m\geq 0$.

A part of this intriguing open problem is the FMP of all $n$-transitive modal logics of the form $\msf{K}\oplus \B^n p\imp \B^{n+1} p$ (cf.~\cite[Problem 11.2]{CZ1997}). The most well-known example of this open problem is perhaps the FMP of the quasi-transitive modal logic $\msf{K}_2^3 = \msf{K}\oplus \B\B p\imp \B\B\B p$ (cf.~e.g.~\cite{KS2016}). We suggest the name `quasi-transitive modal logic' for $\mathsf{K}_2^3$ because frames for it are `almost transitive'. The aim of the present work is to show that $\mathsf{K}_2^3$ has the FMP. 

Our proof proceeds in the algebraic way. It is well-known that, by the duality between modal algebras and frames, the FMP of a normal modal logic $\Lambda$ is equivalent to the {\em finite algebra property} (FAP), i.e., every formula $\varphi$ which is not provable in $\Lambda$ is refuted by a finite $\Lambda$-algebra (cf.~\cite{Yde2007}). In order to show the FAP of the quasi-transitive modal logic $\mathsf{K}_2^3$, we shall prove the FAP of the tense logic $\mathsf{Kt}_2^3$ which is a conservative extension of $\mathsf{K}_2^3$. The tense logic $\mathsf{Kt}_2^3$ extends the minimal tense logic $\msf{Kt}$ (cf.~\cite{BDV2001}) by adding the axiom $\D\D\D p\imp \D\D p$. The core of the proof is a syntactic construction of finite algebra based on a Gentzen sequent calculus $\msf{G}$ for $\mathsf{Kt}_2^3$. We show that every sequent which is not derivable in $\msf{G}$ is refuted by a finite algebra for $\mathsf{Kt}_2^3$. 
There are two main innovative points in the syntactic construction. First, in the definition of sequent calculus $\msf{G}$, a structural operator $\tup{.}$ for the modal operator $\D$ is introduced (cf.~Definition \ref{sequent}).
Second, in order to show that the finite algebra is a $\mathsf{Kt}_2^3$-algebra, a particular form of interpolation lemma for $\msf{G}$ is required (cf.~Lemma \ref{interpolant}). This kind of interpolation lemma was used by Buszkowski \cite{Bus2011} to show the finite embeddability property of residuated algebras.


\section{A Gentzen sequent calculus}

The language of modal logic consists of a denumerable set of propositional variables $\msf{Prop}$, propositional connectives $\bot, \top, \neg, \wedge, \vee$ and a unary modal operator $\D$.
The set of all modal formulas $\mc{L}_\D$ is defined inductively by the following rule:
\[
\mc{L}_\D\ni \varphi::= p\mid \bot \mid\top\mid \neg\varphi\mid (\varphi_1\wedge\varphi_2) \mid (\varphi_1\vee\varphi_2)\mid \D\varphi,~\mrm{where}~p\in\msf{Prop}.
\]
The tense language is the extension of the modal language by a unary modal operator $\bb$. The set of all tense formulas is denoted by $\mathscr{L}_t$. For any number $k\geq 0$, $\D^n\varphi$ is defined by $\D^0\varphi := \varphi$ and $\D^{n+1}\varphi := \D\D^k\varphi$. The {\em complexity} $c(\varphi)$ of a tense formula $\varphi$ is defined inductively as follows:
\begin{align*}
c(p) &=c(\bot)=c(\top)=0. \\
c(\varphi \odot \psi) &= max\{c(\varphi),c(\psi)\}+1,~\text{where $\odot\in\{\wedge, \vee\}$}.\\
c(\#\varphi) & = c(\varphi)+1,~\text{where $\#\in\{\neg, \D,\bb\}$}.
\end{align*}
The {\em tense formula algebra} is denoted by $\mathscr{T}$.

\begin{definition}{\em
A {\em modal algebra} is $\mf{A} = (A, \wedge, \vee, \neg, 0, 1,\D)$ where $(A, \wedge, \vee, \neg, 0, 1)$ is a Boolean algebra, and $\D$ is a unary operator on $A$ with $\D 0 = 0$ and $\D(a\vee b)=\D a\vee \D b$ for all $a,b\in A$.
A modal algebra $\mf{A} = (A, \wedge, \vee, \neg, 0, 1, \D)$ is {\em quasi-transitive}, if $\D\D\D a \leq \D\D a$ for all $a\in A$.

A {\em tense algebra} is $\mf{A} = (A, \wedge, \vee, \neg, 0, 1,\D, \bb)$ where $(A, \wedge, \vee, \neg, 0, 1)$ is a Boolean algebra, and $\D,\bb$ are unary operators on $A$ such that for all $a,b\in A$:
\begin{center}
(Adj) $\D a\leq b$ if and only if $a\leq \bb b$.
\end{center}
A tense algebra $\mf{A} = (A, \wedge, \vee, \neg, 0, 1, \bb, \D)$ is {\em quasi-transitive}, if $\D\D\D a \leq \D\D a$ for all $a\in A$. Let $\mbb{Q}$ be the class of all quasi-transitive algebras.
}
\end{definition}

\begin{lemma}\label{lem:tense:property}
For any quasi-transitive tense algebra $\mf{A}=(A, \wedge, \vee, \neg, 0, 1, \D,\bb)$, the following hold for all $a, b, c\in A$:
\begin{enumerate}
\item[]$(1)$ $\D0=0$ and $\bb 1=1$.
\item[]$(2)$ $a\leq \bb\D a$ and $\D\bb a\leq a$.
\item[]$(3)$ if $a\leq b$, then $\D a\leq \D b$ and $\bb a\leq \bb b$.
\item[]$(4)$ $\D(a\vee b)=\D a\vee\D b$ and $\bb(a\wedge b)=\bb a\wedge\bb b$.
\item[]$(5)$ $\bb\bb a\leq\bb\bb\bb a$.
\end{enumerate}
\end{lemma}
\begin{proof}
We show only $\D(a\vee b)= \D a\vee \D b$ in (5). Other items can be shown easily.
Clearly $a\leq a\vee b$ and $b\leq a\vee b$. By (4), $\D a\leq \D(a\vee b)$ and $\D b\leq \D(a\vee b)$. Then $\D a\vee\D b\leq \D(a\vee b)$. By (2), $a\leq \bb\D a$. By $\D a\leq \D a\vee\D b$ and (3), $\bb\D a\leq \bb(\D a\vee\D b)$. Then $a\leq \bb(\D a\vee\D b)$. Similarly $b\leq \bb(\D a\vee\D b)$. Then $a\vee b\leq \bb(\D a\vee\D b)$. By (Adj), $\D(a\vee b)\leq \D a\vee \D b$. 
\end{proof}

\begin{definition}\label{sequent}{\em
For any tense formula $\varphi$, 
we define $\tup{\varphi}^n$ by induction on $n\geq 0$ as follows: 
\begin{center}
$\tup{\varphi}^0=\varphi$ and $\tup{\varphi}^{n+1}=\tup{\tup{\varphi}^n}$. 
\end{center}
A {\em formula structure} is an expression of the form $\tup{\varphi}^n$ for some tense formula $\varphi$ and natural number $n\geq 0$. Formula structures are denoted by $\Gamma, \Delta$ etc. 
For any set of formulas $X$, the set of all formula structures generated by $X$ is defined as $FS(X)=\{\tup{\varphi}^n\mid \varphi\in X~\&~n\geq 0\}$.

A {\em sequent} is an expression $\Gamma\seq\psi$ where $\Gamma$ is a formula structure and $\psi$ is a tense formula. Sequents are denoted by $s, t$ etc. with or without subscripts. A {\em sequent rule} is a fraction 
\[
\frac{s_1~\ldots~s_n}{s_0}{(R)}
\]
where $s_1,\ldots, s_n$ are called the {\em premisses} and $s_0$ is called the {\em conclusion} of $(R)$.
}
\end{definition}

\begin{definition}{\em
The sequent calculus $\msf{G}$ for the tense logic $\msf{Kt}_2^3$ consists of the following axiom schemata and rules:

$(1)$ Axiom schemata:
\[
(\mathrm{Id})~\varphi\seq\varphi
\quad
(\mathrm{D})~\varphi\wedge(\psi\vee\chi)\seq(\varphi\wedge\psi)\vee(\varphi\wedge\chi)
\quad
(\top)~\varphi\seq\top
\quad
(\bot)~\tup{\bot}^n\seq\psi
\]
\[
(\mrm{LC})~\varphi\wedge\neg\varphi \seq \bot
\quad
(\mrm{EM})~\top\seq\varphi\vee\neg\varphi
\quad
(\D_3^2)~\D^3\varphi\seq\D^2\varphi
\]

$(2)$ Connective rules:
\[
\frac{\tup{\varphi_i}^n\seq\psi}{\tup{\varphi_1\wedge\varphi_2}^n\seq\psi}{(\wedge{\seq})(i=1,2)}
\quad
\frac{\Gamma\seq\psi_1\quad\Gamma\seq\psi_2}{\Gamma\seq\psi_1\wedge\psi_2}{({\seq}\wedge)}
\]
\[
\frac{\tup{\varphi_1}^n\seq\psi\quad \tup{\varphi_2}^n\seq\psi}{\tup{\varphi_1\vee \varphi_2}^n\seq\psi}{(\vee{\seq})}
\quad
\frac{\Gamma\seq\psi_i}{\Gamma\seq\psi_1\vee\psi_2}{({\seq}\vee)(i=1,2)}
\]

$(3)$ Modal rules:
\[
\frac{\tup{\varphi}^{n+1}\seq\psi}{\tup{\D\varphi}^n\seq \psi}{(\D{\seq})}
\quad
\frac{ \Gamma\seq\psi}{\tup{ \Gamma}\seq\D\psi}{({\seq}\D)}
\]
\[
\frac{\tup{\varphi}^n\seq\psi}{\tup{\bb\varphi}^{n+1}\seq \psi}{(\bb{\seq})}
\quad
\frac{\tup{\Gamma}\seq\psi}{\Gamma\seq\bb\psi}{({\seq}\bb)}
\]


$(4)$ Cut rule:
\[
\frac{\Gamma \Rightarrow \varphi \quad \tup{\varphi}^n \Rightarrow \psi}{\tup{\Gamma}^n \Rightarrow \psi}{(Cut)}
\]

In the axiom schema $(\bot)$ and sequent rules, the number $n\geq 0$ is arbitrary.
A {\em derivation} in $\msf{G}$ is a finite tree of sequents $\mc{D}$ in which each node is either an instance of an axiom schema or derived from child node(s) by a sequent rule. The {\em height} of a derivation $\mc{D}$, denoted by $|\mc{D}|$, is the maximal length of branches in $\mc{D}$.
In a derivation, we use $(R)^n$ to denote $n$ times application of the rule $(R)$. A sequent $\Gamma\seq\psi$ is {\em derivable} in $\msf{G}$, notation $\msf{G}\vdash\Gamma\seq\psi$, if there is a derivation $\mc{D}$ in $\msf{G}$ with root node $\Gamma\seq\psi$.
}
\end{definition}

\begin{example}\label{exam}{\em

$(1)$ For $\odot\in\{\wedge,\vee\}$, if $\msf{G}\vdash \varphi_1\seq\psi_1$ and $\msf{G}\vdash \varphi_2\seq\psi_2$, then $\msf{G}\vdash \varphi_1\odot\varphi_2\seq\psi_1\odot\psi_2$. We have the following derivation:
\[
\AxiomC{$\varphi_1\seq\psi_1$}
\RightLabel{\small $(\wedge{\seq})$}
\UnaryInfC{$\varphi_1\wedge\varphi_2\seq\psi_1$}
\AxiomC{$\varphi_2\seq\psi_2$}
\RightLabel{\small $(\wedge{\seq})$}
\UnaryInfC{$\varphi_1\wedge\varphi_2\seq\psi_2$}
\RightLabel{\small $({\seq}\wedge)$}
\BinaryInfC{$\varphi_1\wedge\varphi_2\seq\psi_1\wedge\psi_2$}
\DisplayProof
\]
By a similar derivation, we have $\msf{G}\vdash \varphi_1\vee\varphi_2\seq\psi_1\vee\psi_2$.

$(2)$ If $\msf{G}\vdash\varphi\seq\psi$, then $\msf{G}\vdash\neg\psi\seq\neg\varphi$. Suppose $\msf{G}\vdash\varphi\seq\psi$. Clearly $\msf{G}\vdash\psi\wedge\neg\psi\seq\varphi$. Then we have the following derivation:
\[
\AxiomC{$\varphi\seq\psi$\quad$\neg\psi\seq\neg\psi$}
\RightLabel{\small $(1)$}
\UnaryInfC{$\varphi\wedge\neg\psi\seq\psi\wedge\neg\psi$}
\AxiomC{$\psi\wedge\neg\psi\seq\neg\varphi$}
\RightLabel{\small $(Cut)$}
\BinaryInfC{$\varphi\wedge\neg\psi\seq\neg\varphi$}
\AxiomC{$\neg\varphi\seq\neg\varphi$}
\RightLabel{\small $(\wedge{\seq})$}
\UnaryInfC{$\neg\varphi\wedge\neg\psi\seq\neg\varphi$}
\RightLabel{\small $(\vee{\seq})$}
\BinaryInfC{$(\varphi\wedge\neg\psi)\vee(\neg\varphi\wedge\neg\psi)\seq\neg\varphi$}
\DisplayProof
\]
By $(\mrm{D})$, $\msf{G}\vdash (\varphi\vee\neg\varphi)\wedge\neg\psi\seq(\varphi\wedge\neg\psi)\vee(\neg\varphi\wedge\neg\psi)$.
By $(Cut)$, $\msf{G}\vdash (\varphi\vee\neg\varphi)\wedge\neg\psi\seq\neg\varphi$. By $(\top)$ and $(1)$, $\msf{G}\vdash \top\wedge\neg\psi\seq(\varphi\vee\neg\varphi)\wedge\neg\psi$. Clearly $\msf{G}\vdash \neg\psi\seq\top\wedge\neg\psi$. By $(Cut)$, $\msf{G}\vdash \neg\psi\seq\neg\varphi$.

$(3)$ If $\msf{G}\vdash\varphi\seq\psi$, then $\msf{G}\vdash\D\varphi\seq\D\psi$ and $\msf{G}\vdash\bb\varphi\seq\bb\psi$. Suppose $\msf{G}\vdash\varphi\seq\psi$. We have the following derivations:
\[
\AxiomC{$\varphi\seq\psi$}
\RightLabel{\small $({\seq}\D)$}
\UnaryInfC{$\tup{\varphi}\seq\D\psi$}
\RightLabel{\small $(\D{\seq})$}
\UnaryInfC{${\D\varphi}\seq\D\psi$}
\DisplayProof
\quad\quad
\AxiomC{$\varphi\seq\psi$}
\RightLabel{\small $(\bb{\seq})$}
\UnaryInfC{$\tup{\bb\varphi}\seq\psi$}
\RightLabel{\small $({\seq}\bb)$}
\UnaryInfC{${\bb\varphi}\seq\bb\psi$}
\DisplayProof
\]

$(4)$ For any $n\geq 0$, $\msf{G}\vdash \tup{\varphi}^n\seq\D^n\varphi$. The case $n=0$ is trivial. Let $n>0$. Starting from the axiom $\varphi\seq\varphi$, by $n$ times application of $({\seq}\D)$, we have $\msf{G}\vdash \tup{\varphi}^n\seq\D^n\varphi$.

$(5)$ $\msf{G}\vdash \D(\varphi\vee\psi)\seq\D\varphi\vee\D\psi$. We have the following derivation:
\[
\AxiomC{$\tup{\varphi}\seq\D\varphi$}
\RightLabel{$({\seq}\vee)$}
\UnaryInfC{$\tup{\varphi}\seq\D\varphi\vee\D\psi$}
\AxiomC{$\tup{\psi}\seq\D\psi$}
\RightLabel{$({\seq}\vee)$}
\UnaryInfC{$\tup{\psi}\seq\D\varphi\vee\D\psi$}
\RightLabel{$(\vee{\seq})$}
\BinaryInfC{$\tup{\varphi\vee\psi}\seq\D\varphi\vee\D\psi$}
\RightLabel{$(\D{\seq})$}
\UnaryInfC{$\D(\varphi\vee\psi)\seq\D\varphi\vee\D\psi$}
\DisplayProof
\]

$(6)$ $\msf{G}\vdash\D\bb\varphi\seq\varphi$. Starting from $\varphi\seq\varphi$, using $(\bb{\seq})$ and $(\D{\seq})$, we have $\msf{G}\vdash\D\bb\varphi\seq\varphi$.
}
\end{example}

\begin{definition}{\em
For any quasi-transitive tense algebra $\mf{A}=(A, \wedge, \vee, \neg, 0, 1, \D, \bb)$, an {\em assignment} in $\mf{A}$ is a 
function $\theta: \msf{Prop}\imp A$. Let $\w{\theta}:\mathscr{L}_t\imp A$ be the homomorphic extension of $\theta$ to the tense formula algebra $\mf{T}$. For any formula structure $\tup{\varphi}^n$, we define $\tau(\tup{\varphi}^n) = \D^n\varphi$. A sequent $\Gamma\seq\psi$ is {\em valid} in $\mbb{Q}$, notation $\Gamma\models_\mbb{Q}\psi$, if $\w{\theta}(\tau(\Gamma))\leq \w{\theta}(\psi)$ for any quasi-transitive tense algebra $\mf{A}$ and assignment $\theta$ in $\mf{A}$.
A sequent rule with premisses $s_1,\ldots,s_n$ and conclusion $s_0$ {\em preserves validity in $\mbb{Q}$}, if $s_0$ is valid in $\mbb{Q}$ whenever $s_i$ for all $1\leq i\leq n$ are valid in $\mbb{Q}$.
}
\end{definition}

A formula $\varphi$ is {\em equivalent} to $\psi$ with respect to $\msf{G}$, notation $\varphi\sim_\msf{G}\psi$, if $\msf{G}\vdash\varphi\seq\psi$ and $\msf{G}\vdash\psi\seq\varphi$. Let $|\varphi|_\msf{G}=\{\chi\in\mathscr{L}\mid \varphi\sim_\msf{G}\chi\}$. For any set of formulas $T$, let $|T|_\msf{G}=\{|\varphi|_\msf{G}\mid \varphi\in T\}$.

\begin{lemma}
The relation $\sim_t$ is a congruence relation on the tense formula algebra $\mathscr{T}$. 
\end{lemma}
\begin{proof}
It is trivial that $\sim_t$ is an equivalence relation on $\mathscr{L}_t$. Suppose $\varphi_1\sim_t \psi_1$ and $\varphi_2\sim_t \psi_2$. By Example \ref{exam} $(1)$,
$\varphi_1\odot\varphi_2\sim_t\psi_1\odot\psi_2$ for $\odot\in\{\wedge,\vee\}$. Suppose $\varphi\sim_t \psi$. By Example \ref{exam} $(2)$ and $(3)$, $\#\varphi\sim_t\#\psi$ for $\#\in\{\neg,\D,\bb\}$. 
\end{proof}

Let $\mf{L}_t$ be the quotient algebra of the tense formula algebra $\mathscr{T}$ under $\sim_t$. One can easily show that $\mf{L}_t$ is quasi-transitive. Moreover, for any $\varphi,\psi\in\mathscr{L}_t$, $|\varphi|_\msf{G}\leq|\psi|_\msf{G}$ if and only if $\msf{G}\vdash\varphi\seq\psi$.

\begin{theorem}
For any sequent $\Gamma\seq\psi$,
$\msf{G}\vdash\Gamma\seq\psi$ if and only if $\mbb{Q}\models \Gamma\seq\psi$.
\end{theorem}
\begin{proof}
Suppose $\msf{G}\vdash\Gamma\seq\psi$. All axioms are obviously valid in $\mbb{Q}$. One can show that all sequent rules of $\msf{G}$ preserve validity in 
$\mbb{Q}$. For $(Cut)$, assume $\Gamma\models_\mbb{Q}\varphi$ and $\tup{\varphi}^n\models_\mbb{Q}\psi$. Let $\mf{A}$ be any algebra in $\mbb{Q}$ and $\theta$ be an assignment in $\mf{A}$.
Then $\w{\theta}(\tau(\Gamma))\leq \w{\theta}(\varphi)$ and 
$\w{\theta}(\D^n\varphi)\leq\w{\theta}(\psi)$. Then $\D^n\w{\theta}(\tau(\Gamma))\leq \D^n\w{\theta}(\varphi)$ and 
$\D^n\w{\theta}(\varphi)\leq\w{\theta}(\psi)$. Hence $\D^n\w{\theta}(\tau(\Gamma))\leq \w{\theta}(\psi)$.
Then $\tup{\Gamma}^n\models_\mbb{Q} \psi$. The other cases can be shown easily. Hence $\mbb{Q}\models \Gamma\seq\psi$.
Suppose $\msf{G}\not\vdash\Gamma\seq\psi$. By $\msf{G}\vdash \Gamma\seq\tau(\Gamma)$ and $(Cut)$, $\msf{G}\not\vdash \tau(\Gamma)\seq\psi$. Hence $|\tau(\Gamma)|\not\leq |\psi|$. Let $\theta$ be the assignment in $\mf{L}_t$ with ${\theta}(p)=|p|$ for each $p\in\msf{Prop}$. One can easily show by induction on the complexity of $\varphi$ that $\w{\theta}(\varphi)=|\varphi|$. Then $\w{\theta}(\tau(\Gamma))\not\leq\w{\theta}(\psi)$. Hence $\mbb{Q}\not\models\Gamma\seq\psi$.
\end{proof}

\section{Finite Model Property}
In this section, we prove the finite algebra property (FAP) of the sequent calculus $\msf{G}$, i.e., if $\msf{G}\not\vdash\Gamma\seq\psi$, there is a finite pretransitive tense algebra that refutes $\Gamma\seq\psi$. The FMP of $\msf{Kt}_2^3$ and $\msf{K}_2^3$ is derived from the FAP.

\begin{definition}\label{def:T}{\em 
For any set of tense formulas $X$ with $\top,\bot\in X$, the sets $X^b$ and $X^\D$ are defined as follows:
\begin{itemize}
\item $X^b$ is the smallest set of tense formulas such that $X\sub X^b$ and $X^b$ is closed under the operations $\neg$, $\wedge$ and $\vee$.
\item $X^\D=\{\D^k \varphi\mid \varphi\in X~\&~0\leq k\leq 3\}$.
\end{itemize}
For any finite set of tense formulas $T$ with $\top,\bot\in T$, let $T^\circ = (T^\D)^b$ and $T^\bu = T^\circ\setminus T^\D$.
}
\end{definition}

In this section, we stipulate that $T$ is a finite set of tense formulas with $\top,\bot\in T$. Obviously $T\sub T^\D$.
A sequent $\Gamma\seq \psi$ is $T^\circ$-{\em derivable} in $\msf{G}$, notation $\msf{G}\vdash\Gamma\seq_{T^\circ}\psi$, if there is a derivation $\mc{D}$ of $\Gamma\seq \psi$ in $\msf{G}$ such that all formulas in $\mc{D}$ belong to $T^\circ$. 

\begin{lemma}\label{rule:R23}
For any $n\geq 0$, if $\msf{G}\vdash \tup{\varphi}^{n+2}\seq_{T^\circ}\psi$ and $\varphi\in T^\D$, then $\msf{G}\vdash\tup{\varphi}^{n+3}\seq_{T^\circ}\psi$.
\end{lemma}
\begin{proof}
Assume $\msf{G}\vdash\tup{\varphi}^{n+2}\seq_{T^\circ}\psi$. Let $\varphi=\D^k\chi$ for some formula $\chi\in T$ where $k\geq 0$. Clearly $\msf{G}\vdash \tup{\chi}^k\seq_{T^\circ}\D^k\chi$.
We have the following derivation:
\[
\AxiomC{$\chi\seq_{T^\circ}\chi$}
\RightLabel{$({\seq}\D)^3$}
\UnaryInfC{$\tup{\chi}^3\seq_{T^\circ}\D^3\chi$}
\AxiomC{$\D^3\chi\seq_{T^\circ}\D^2\chi$}
\AxiomC{$\tup{\chi}^k\seq_{T^\circ}\D^k\chi$\quad$\tup{\D^k\chi}^{n+2}\seq_{T^\circ}\psi$}
\RightLabel{$(Cut)$}
\UnaryInfC{$\tup{\chi}^{k+n+2}\seq_{T^\circ}\psi$}
\RightLabel{$(\D{\seq})^2$}
\UnaryInfC{$\tup{\D^2\chi}^{k+n}\seq_{T^\circ}\psi$}
\RightLabel{$(Cut)$}
\BinaryInfC{$\tup{\D^3\chi}^{k+n}\seq_{T^\circ}\psi$}
\RightLabel{$(Cut)$}
\BinaryInfC{$\tup{\chi}^{k+n+3}\seq_{T^\circ} \psi$}
\RightLabel{$(\D{\seq})^{k}$}
\UnaryInfC{$\tup{\D^k\chi}^{n+3}\seq_{T^\circ} \psi$}
\DisplayProof
\]
Note that $\D^3\chi, \D^2\chi\in T^\circ$ since $\chi\in T$.
Hence $\msf{G}\vdash\tup{\varphi}^{n+3}\seq_{T^\circ}\psi$.
\end{proof}

\begin{lemma}[Interpolation]\label{interpolant}
For any set of tense formulas $T$ and $n\geq m\geq 0$, if $\msf{G}\vdash\tup{\varphi}^n\seq_{T^\circ}\psi$, there exists a formula $\gamma\in T^\circ$ with $\msf{G}\vdash\tup{\varphi}^m\seq_{T^\circ}\gamma$ and $\msf{G}\vdash\tup{\gamma}^{n-m}\seq_{T^\circ}\psi$.
\end{lemma}
\begin{proof}
The required formula $\gamma\in T^\circ$ is called an interpolant. If $m=n$, we choose $\psi\in T^\circ$ as a required interpolant. If $m=0$, we choose $\varphi\in T^\circ$ as a required interpolant. Let $n>m>0$. Assume $\msf{G}\vdash\tup{\varphi}^n\seq_{T^\circ}\psi$. 
There is a derivation $\mc{D}$ of $\tup{\varphi}^n\seq_{T^\circ}\psi$ in $\msf{G}$. The proof proceeds by induction on the height $|\mc{D}|$. Suppose $|\mc{D}|=0$. Then $\varphi=\bot$. We choose $\bot$ as a required interpolant. Suppose $|\mc{D}|>0$. 
Let $\tup{\varphi}^n\seq_{T^\circ}\psi$ be obtained by a rule $(R)$.

(1) $(R)$ is a connective rule. One can get the required interpolant by induction hypothesis and the rule $(R)$. We have the following cases:

(1.1) Let $(R)$ be $(\wedge{\seq})$ and the derivation end with 
\[
\frac{\tup{\varphi_i}^n\seq_{T^\circ}\psi}{\tup{\varphi_1\wedge \varphi_2}^n\seq_{T^\circ}\psi}{(\wedge{\seq})}
\]
where $\varphi=\varphi_1\wedge\varphi_2$ and $i=1,2$.
By induction hypothesis, there is a formula $\gamma\in T^\circ$ with
(i) $\msf{G}\vdash\tup{\varphi_i}^m\seq_{T^\circ}\gamma$ and (ii) $\msf{G}\vdash\tup{\gamma}^{n-m}\seq_{T^\circ}\psi$.
By (i) and $(\wedge{\seq})$, $\msf{G}\vdash\tup{\varphi_1\wedge\varphi_2}^m\seq_{T^\circ}\gamma$. 
Hence $\gamma$ is a required interpolant.

(1.2) Let $(R)$ be $({\seq}\wedge)$ and the derivation end with 
\[
\frac{\tup{\varphi}^n\seq_{T^\circ}\psi_1\quad \tup{\varphi}^n\seq_{T^\circ}\psi_2}{\tup{\varphi}^n\seq_{T^\circ}\psi_1\wedge\psi_2}{({\seq}\wedge)}
\]
By induction hypothesis, there are formulas $\chi_1, \chi_2\in T^\circ$ with
(i) $\msf{G}\vdash\tup{\varphi}^m\seq_{T^\circ}\chi_1$;
(ii) $\msf{G}\vdash\tup{\chi_1}^{n-m}\seq_{T^\circ}\psi_1$;
(iii) $\msf{G}\vdash\tup{\varphi}^m\seq_{T^\circ}\chi_2$;
(iv) $\msf{G}\vdash\tup{\chi_2}^{n-m}\seq_{T^\circ}\psi_2$.
By (i) and (iii), using $({\seq}\wedge)$, we have $\msf{G}\vdash\tup{\varphi}^m\seq_{T^\circ}\chi_1\wedge\chi_2$. 
Since $\msf{G}\vdash\chi_1\wedge\chi_2\seq_{T^\circ}\chi_1$ and $\msf{G}\vdash\chi_1\wedge\chi_2\seq_{T^\circ}\chi_2$, by (ii) and (iv),
using $(Cut)$, we have $\msf{G}\vdash\tup{\chi_1\wedge\chi_2}^{n-m}\seq_{T^\circ}\psi_1$ and $\msf{G}\vdash\tup{\chi_1\wedge\chi_2}^{n-m}\seq_{T^\circ}\psi_2$. By $({\seq}\wedge)$, $\msf{G}\vdash\tup{\chi_1\wedge\chi_2}^{n-m}\seq_{T^\circ}\psi_1\wedge\psi_2$.
Note that $\chi_1\wedge\chi_2\in T^\circ$ since $\chi_1,\chi_2\in T^\circ$. Hence $\chi_1\wedge\chi_2$ is a required interpolant.

(1.3) Let $(R)$ be $(\vee{\seq})$ and the derivation end with 
\[
\frac{\tup{\varphi_1}^n\seq_{T^\circ}\psi\quad \tup{\varphi_2}^n\seq_{T^\circ}\psi}{\tup{\varphi_1\vee \varphi_2}^n\seq_{T^\circ}\psi}{(\vee{\seq})}
\]
where $\varphi=\varphi_1\vee\varphi_2$.
By induction hypothesis, there are formulas $\gamma_1, \gamma_2\in T^\circ$ with
(i) $\msf{G}\vdash\tup{\varphi_1}^m\seq_{T^\circ}\gamma_1$;
(ii) $\msf{G}\vdash\tup{\gamma_1}^{n-m}\seq_{T^\circ}\psi$;
(iii) $\msf{G}\vdash\tup{\varphi_2}^m\seq_{T^\circ}\gamma_2$;
(iv) $\msf{G}\vdash\tup{\gamma_2}^{n-m}\seq_{T^\circ}\psi$.
By (i) and (iii), applying $({\seq}\vee)$ and $(\vee{\seq})$, we have $\msf{G}\vdash\tup{\varphi_1\vee\varphi_2}^m\seq_{T^\circ}\gamma_1\vee\gamma_2$. By (ii) and (iv), using $({\seq}\vee)$, $\msf{G}\vdash\tup{\gamma_1\vee\gamma_2}^{n-m}\seq_{T^\circ}\psi$. Note that $\gamma_1\vee\gamma_2\in T^\circ$ since $\gamma_1,\gamma_2\in T^\circ$. Hence $\chi_1\vee\chi_2$ is a required interpolant.

(1.4) Let $(R)$ be $({\seq}\vee)$ and the derivation end with 
\[
\frac{\tup{\varphi}^n\seq_{T^\circ}\psi_i}{\tup{\varphi}^n\seq_{T^\circ}\psi_1\vee\psi_2}{({\seq}\vee)}
\]
where $\psi=\psi_1\vee\psi_2$ and $i=1,2$.
By induction hypothesis, there is a formula $\chi\in T^\circ$ with
(i) $\msf{G}\vdash\tup{\varphi}^m\seq_{T^\circ}\chi$ and
(ii) $\msf{G}\vdash\tup{\chi}^{n-m}\seq_{T^\circ}\psi_i$.
By (ii) and $({\seq}\vee)$, we have $\msf{G}\vdash\tup{\chi}^{n-m}\seq_{T^\circ}\psi_1\vee\psi_2$. Hence $\chi$ is a required interpolant.

(2) $(R)$ is a modal rule. We have the following cases:

(2.1) $(R)$ is $(\D{\seq})$. Let the derivation end with
\[
\frac{\tup{\chi}^{n+1}\seq_{T^\circ}\psi}{\tup{\D\chi}^n\seq_{T^\circ} \psi}{(\D{\seq})}
\]
where $\varphi=\D\chi$. By induction hypothesis, (i) $\msf{G}\vdash\tup{\chi}^{m+1}\seq_{T^\circ}\gamma$ and (ii) $\msf{G}\vdash\tup{\gamma}^{n-m}\seq_{T^\circ}\psi$ for some $\gamma\in T^\circ$.
By (i) and $(\D{\seq})$, $\msf{G}\vdash\tup{\D\chi}^{m}\seq_{T^\circ}\gamma$. Then $\gamma$ is a required interpolant.

(2.2) $(R)$ is $({\seq}\D)$. Let the derivation end with
\[
\frac{\tup{\varphi}^{n-1}\seq\chi}{\tup{\varphi}^n\seq\D\chi}{({\seq}\D)}
\]
where $\psi=\D\chi$. By induction hypothesis, (i) $\msf{G}\vdash\tup{\varphi}^{m}\seq_{T^\circ}\gamma$ and (ii) $\msf{G}\vdash\tup{\gamma}^{n-m-1}\seq_{T^\circ}\chi$ for some $\gamma\in T^\circ$.
By (ii) and $({\seq}\D)$, $\msf{G}\vdash\tup{\gamma}^{n-m}\seq_{T^\circ}\D\chi$. Then $\gamma$ is a required interpolant.

(2.3) $(R)$ is $(\bb{\seq})$. Let the derivation end with
\[
\frac{\tup{\chi}^{n-1}\seq\psi}{\tup{\bb\chi}^{n}\seq \psi}{(\bb{\seq})}
\]
where $\varphi=\bb\chi$. By induction hypothesis, (i) $\msf{G}\vdash\tup{\chi}^{m-1}\seq_{T^\circ}\gamma$ and (ii) $\msf{G}\vdash\tup{\gamma}^{n-m}\seq_{T^\circ}\psi$ for some $\gamma\in T^\circ$.
By (i) and $(\bb{\seq})$, $\msf{G}\vdash\tup{\bb\chi}^{m}\seq_{T^\circ}\gamma$. Then $\gamma$ is a required interpolant.

(2.4) $(R)$ is $({\seq}\bb)$. Let the derivation end with
\[
\frac{\tup{\varphi}^{n+1}\seq\chi}{\tup{\varphi}^{n}\seq\bb\chi}{({\seq}\bb)}
\]
where $\psi=\bb\chi$. By induction hypothesis, (i) $\msf{G}\vdash\tup{\varphi}^{m}\seq_{T^\circ}\gamma$ and (ii) $\msf{G}\vdash\tup{\gamma}^{n-m+1}\seq_{T^\circ}\chi$ for some $\gamma\in T^\circ$.
By (ii) and $({\seq}\bb)$, $\msf{G}\vdash\tup{\gamma}^{n-m}\seq_{T^\circ}\bb\chi$. Then $\gamma$ is a required interpolant.

(3) $(R)$ is $(Cut)$. Let the derivation end with
\[
\frac{\tup{\varphi}^i\seq_{T^\circ}\chi\quad\tup{\chi}^j\seq_{T^\circ}\psi}{\tup{\varphi}^n\seq_{T^\circ}\psi}{(Cut)}
\]
where $i+j=n$. Suppose $m\leq i$. By induction hypothesis, (i) $\msf{G}\vdash\tup{\varphi}^{m}\seq_{T^\circ}\gamma$ and (ii) $\msf{G}\vdash\tup{\gamma}^{i-m}\seq_{T^\circ}\chi$ for some $\gamma\in T^\circ$. By the right premiss of $(Cut)$ and (ii), using $(Cut)$,  $\msf{G}\vdash\tup{\gamma}^{n-m}\seq_{T^\circ}\psi$. Then $\gamma$ is a required interpolant. Suppose $m>i$. By induction hypothesis, (iii) $\msf{G}\vdash\tup{\chi}^{m-i}\seq_{T^\circ}\gamma'$ and (iv) $\msf{G}\vdash\tup{\gamma'}^{n-m}\seq_{T^\circ}\chi$ for some $\gamma'\in T^\circ$. By the left premiss of $(Cut)$ and (iii), using $(Cut)$, $\msf{G}\vdash\tup{\varphi}^{m}\seq_{T^\circ}\gamma'$. Then $\gamma'$ is a required interpolant.
\end{proof}

\begin{definition}{\em
Let $FS(T^\D), FS(T^\circ)$ and $FS(T^\bu)$ be sets of all formula structures generated by $T^\D, T^\circ$ and $T^\bu$ respectively. For any $\varphi\in T^\circ$, we define
\[
G(\varphi) = \{\tup{\chi}^n\in FS(T^\D)\mid  \msf{G}\vdash\tup{\chi}^n\seq_{T^\circ}\varphi\}~\text{and}~
\den{\varphi} = G(\varphi) \cup FS(T^\bu).
\]
Let $\den{T^\circ} = \{\den{\varphi}\mid \varphi\in T^\circ\}$. We define operations $\bot^\circ, \top^\circ, \neg^\circ, \wedge^\circ$ and $\vee^\circ$ on $\den{T^\circ}$ as follows:
\[
\top^\circ = \den{\top},\quad
\bot^\circ = \den{\bot},\quad
\neg^\circ\den{\varphi} = \den{\neg\varphi},\quad
\den{\varphi}\wedge^\circ\den{\psi} = \den{\varphi\wedge\psi},\quad
\den{\varphi}\vee^\circ\den{\psi} = \den{\varphi\vee\psi}.
\]
Let $\mf{B}(T^\circ) = (\den{T^\circ}, \wedge^\circ, \vee^\circ, \neg^\circ, \top^\circ, \bot^\circ)$. The binary relation $\leq^\circ$ on $\den{T^\circ}$ is defined as follows: $\den{\varphi}\leq^\circ \den{\psi}$ if and only if $\den{\varphi}\wedge^\circ\den{\psi}=\den{\varphi}$.
}\end{definition}

\begin{lemma}\label{lem:T:inclusion}
For any $\varphi,\psi\in T^\circ$, the following hold:

$(1)$ $\den{\varphi}\sub\den{\psi}$ if and only if $G(\varphi)\sub G(\psi)$.

$(2)$ if $\msf{G}\vdash \varphi\seq_{T^\circ}\psi$, then $\den{\varphi}\sub\den{\psi}$.

$(3)$ $\den{T^\circ}$ is finite.

$(4)$ $\den{\varphi}\leq^\circ \den{\psi}$ if and only if $\den{\varphi}\sub \den{\psi}$.

$(5)$ $\den{\varphi\wedge\psi}=\den{\varphi}\cap \den{\psi}$.
\end{lemma}
\begin{proof}
$(1)$ Assume $\den{\varphi}\sub\den{\psi}$. Let $\tup{\chi}^n\in G(\varphi)\sub \den{\varphi}$. Then $\chi\in T^\D$ and $\tup{\chi}^n\in\den{\psi}$. Then $\tup{\chi}^n\in G(\psi)$. Hence $G(\varphi)\sub G(\psi)$. Assume $G(\varphi)\sub G(\psi)$.
Then $G(\varphi)\cup FS(T^\bu) \sub G(\psi)\cup FS(T^\bu)$, i.e., 
$\den{\varphi}\sub\den{\psi}$.

$(2)$ Assume $\msf{G}\vdash \varphi\seq_{T^\circ}\psi$. Suppose $\tup{\chi}^n\in G(\varphi)$. Then $\chi\in T^\D$ and $\msf{G}\vdash\tup{\chi}^n\seq_{T^\circ}\varphi$. By $(Cut)$, $\msf{G}\vdash\tup{\chi}^n\seq_{T^\circ}\psi$. Then 
$\tup{\chi}^n\in G(\psi)$. Then $G(\varphi)\sub G(\psi)$. By $(1)$, $\den{\varphi}\sub\den{\psi}$.

$(3)$ Since $T$ is finite, $T^\D$ is also finite. There are only finitely many non-equivalent formulas in $(T^\D)^b$. Then  $|T^\circ|_\msf{G}$ is finite. By $(2)$, if $\varphi\sim_\msf{G}\psi$, then $\den{\varphi}=\den{\psi}$. Hence $\den{T^\circ}$ is finite.

$(4)$ Assume $\den{\varphi}\leq^\circ \den{\psi}$. Then $\den{\varphi}\wedge^\circ\den{\psi}=\den{\varphi\wedge\psi}=\den{\varphi}$.
Suppose $\tup{\chi}^n\in G(\varphi)$. Then $\chi\in T^\D$ and $\msf{G}\vdash\tup{\chi}^n\seq_{T^\circ}\varphi$. 
Since $\den{\varphi\wedge\psi}=\den{\varphi}$, we have $\tup{\chi}^n\in G(\varphi\wedge\psi)$. Then
$\msf{G}\vdash\tup{\chi}^n\seq_{T^\circ}\varphi\wedge\psi$. Clearly $\msf{G}\vdash\varphi\wedge\psi\seq_{T^\circ} \psi$. By $(Cut)$,  $\msf{G}\vdash\tup{\chi}^n\seq_{T^\circ}\psi$. Then $\tup{\chi}^n\in G(\psi)$. Hence $G(\varphi)\sub G(\psi)$. By $(1)$, we have $\den{\varphi}\sub \den{\psi}$.

Assume $\den{\varphi}\sub \den{\psi}$. By $(1)$, we have $G(\varphi)\sub G(\psi)$. Suppose $\tup{\chi}^n\in G(\varphi)$. Then $\tup{\chi}^n\in G(\psi)$. Then $\msf{G}\vdash\tup{\chi}^n\seq_{T^\circ}\varphi$ and $\msf{G}\vdash\tup{\chi}^n\seq_{T^\circ}\psi$. By $({\seq}\wedge)$, $\msf{G}\vdash\tup{\chi}^n\seq_{T^\circ}\varphi\wedge\psi$. Then $\tup{\chi}^n\in G(\varphi\wedge\psi)$. Hence $G(\varphi)\sub G(\varphi\wedge\psi)$. By $(1)$, $\den{\varphi}\sub\den{\varphi\wedge\psi}$. 
Suppose $\tup{\chi}^n\in G(\varphi\wedge\psi)$.
Then $\msf{G}\vdash \tup{\chi}^n\seq_{T^\circ}\varphi\wedge\psi$. Clearly $\msf{G}\vdash \varphi\wedge\psi\seq_{T^\circ}\varphi$. By $(Cut)$, $\msf{G}\vdash \tup{\chi}^n\seq_{T^\circ}\varphi$. Then $\tup{\chi}^n\in G(\varphi)$.
By $(1)$, $\den{\varphi\wedge\psi}\sub\den{\varphi}$. Hence $\den{\varphi\wedge\psi}=\den{\varphi}$, i.e., $\den{\varphi}\leq^\circ\den{\psi}$.

$(5)$ Clearly $\msf{G}\vdash \varphi\wedge\psi\seq_{T^\circ}\varphi$ and $\msf{G}\vdash \varphi\wedge\psi\seq_{T^\circ}\psi$. By (2), $\den{\varphi\wedge\psi}\sub \den{\varphi}$ and $\den{\varphi\wedge\psi}\sub \den{\psi}$. Then $\den{\varphi\wedge\psi}\sub \den{\varphi}\cap\den{\psi}$. 
Clearly $\den{\varphi}\cap\den{\psi} = (G(\varphi)\cap G(\psi)) \cup FS(T^\bu)$. Suppose 
$\tup{\chi}^n\in G(\varphi)\cap G(\psi)$. Then $\msf{G}\vdash\tup{\chi}^n\seq_{T^\circ}\varphi$ and $\msf{G}\vdash\tup{\chi}^n\seq_{T^\circ}\psi$. By $({\seq}\wedge)$, $\msf{G}\vdash\tup{\chi}^n\seq_{T^\circ}\varphi\wedge\psi$.Then $\tup{\chi}^n\in G(\varphi\wedge\psi)$. Then $G(\varphi)\cap G(\psi)\sub G(\varphi\wedge\psi)$.
Then $\den{\varphi}\cap\den{\psi}\sub\den{\varphi\wedge\psi}$.
\end{proof}

\begin{lemma}\label{lem:boole}
$\mf{B}(T^\circ)$ is a finite Boolean algebra.
\end{lemma}
\begin{proof}
By Lemma \ref{lem:T:inclusion} $(3)$, $\den{T^\circ}$ is finite. It is easy to show that $(\den{T^\circ}, \wedge^\circ, \vee^\circ)$ is a distributive lattice. Here we show only the law of distributivity. Suppose $\varphi,\psi,\chi\in T^\circ$. Then $\den{\varphi}\wedge^\circ(\den{\psi}\vee^\circ\den{\chi}) = \den{\varphi\wedge(\psi\vee\chi)}$ and $(\den{\varphi}\wedge^\circ\den{\psi})\vee^\circ(\den{\varphi}\wedge^\circ\den{\chi}) = \den{(\varphi\wedge\psi)\vee(\varphi\wedge\chi)}$. Clearly $\msf{G}\vdash\varphi\wedge(\psi\vee\chi)\seq_{T^\circ}(\varphi\wedge\psi)\vee(\varphi\wedge\chi)$.
Then $G(\varphi\wedge(\psi\vee\chi))\sub G((\varphi\wedge\psi)\vee(\varphi\wedge\chi))$.
By Lemma \ref{lem:T:inclusion} $(1)$ and $(4)$, $\den{\varphi}\wedge^\circ(\den{\psi}\vee^\circ\den{\chi})\leq^\circ (\den{\varphi}\wedge^\circ\den{\psi})\vee^\circ(\den{\varphi}\wedge^\circ\den{\chi})$.
By Lemma \ref{lem:T:inclusion} $(4)$, the order $\leq^\circ$ is equal to $\sub$. By $\msf{G}\vdash \varphi\seq_{T^\circ}\top$ and $\msf{G}\vdash \bot\seq_{T^\circ}\varphi$, using Lemma \ref{lem:T:inclusion} $(2)$ and $(4)$, $\den{\varphi}\leq^\circ \den{\top}$ and $\den{\bot}\leq^\circ \den{\varphi}$. Hence the $(\wedge^\circ,\vee^\circ,\top^\circ,\bot^\circ)$-reduct of $\mf{B}(T^\circ)$ is a bounded distributive lattice. By $\msf{G}\vdash \varphi\wedge\neg\varphi\seq_{T^\circ}\bot$ and $\msf{G}\vdash \top\seq_{T^\circ}\varphi\vee\neg\varphi$, using Lemma \ref{lem:T:inclusion} $(2)$ and $(4)$, $\den{\varphi}\wedge^\circ\den{\neg\varphi} =\den{\varphi\wedge\neg\varphi} \leq^\circ\den{\bot}$ and $\den{\top}\leq^\circ \den{\varphi\vee\neg\varphi}=\den{\varphi}\vee^\circ\den{\neg\varphi}$. Hence $\mf{B}(T^\circ)$ is a Boolean algebra.
\end{proof}

\begin{definition}
The operation $C:\mc{P}(FS(T^\circ))\imp \mc{P}(FS(T^\circ))$ is defined as follows:
\[
C(X) = \bigcap\{\den{\varphi}\mid X\sub \den{\varphi}\in \den{T^\circ}\}.
\]
The unary operations $\D$ and $\bb$ on $\mc{P}(FS(T^\circ))$ are defined as follows:
\[
\D X = \{\tup{\Gamma} \mid \Gamma\in X\}
~\text{and}~
\bb X = \{\Gamma\mid \tup{\Gamma}\in X\}.
\]
The unary operation $\D_c$ on $\mc{P}(FS(T^\circ))$ is defined by $\D_c X = C(\D X)$.
\end{definition}

\begin{lemma}\label{lem:C1}
For any $X, Y\in \mc{P}(FS(T^\circ))$, the following hold:

$(1)$ $\den{\top}=FS(T^\circ)$.

$(2)$ $C(X) =\den{\varphi}$ for some formula $\varphi\in T^\circ$.

$(3)$ $X\sub C(X)$.

$(4)$ if $X\sub Y$, then $C(X)\sub C(Y)$.

$(5)$ $C(C(X))\sub C(X)$.

$(6)$ $C(\den{\varphi})=\den{\varphi}$.
\end{lemma}
\begin{proof}
$(1)$ Clearly $FS(T^\circ)=FS(T^\D)\cup FS(T^\bu)$ and $\den{\top}\sub FS(T^\circ)$.
It is obvious that $\msf{G}\vdash\tup{\chi}^n\seq\top$ for any $n\geq 0$ and $\chi\in T^\D$. Then $FS(T^\D)\sub\den{\top}$.
Clearly $FS(T^\bu)\sub \den{\top}$. Then $FS(T^\circ)\sub \den{\top}$. Hence $\den{\top}=FS(T^\circ)$. 

$(2)$ Let $\mf{X} = \{\den{\psi}\mid X\sub \den{\psi}\in \den{T^\circ}\}$. By Lemma \ref{lem:T:inclusion} $(3)$, $\den{T^\circ}$ is finite. Then $\mf{X}$ is finite. By (1), $\mf{X}\neq \ve$. Let $\mf{X}=\{\den{\psi_0},\ldots,\den{\psi_n}\}$. Then $C(X)=\den{\psi_0}\cap\ldots\cap\den{\psi_n}$.
By Lemma \ref{lem:T:inclusion} $(5)$, $C(X)=\den{\bigwedge_{i\leq n}\psi_i}$. Clearly $\varphi=\bigwedge_{i\leq n}\psi_i\in T^\circ$. 

$(3)$ It follows from the definition of $C(X)$.

$(4)$ Assume $X\sub Y$. By $(2)$, let $C(Y)=\den{\psi}$ for some $\psi\in T^\circ$. By $(3)$, $Y\sub \den{\psi}$. Then $X\sub \den{\psi}$. By the definition of $C$, we have $C(X)\sub \den{\psi} = C(Y)$. 

$(5)$ By $(2)$, let $C(X)=\den{\chi}$ for some $\chi\in T^\circ$. Then $C(C(X))\sub \den{\chi}=C(X)$.

$(6)$ By $(3)$, $\den{\varphi}\sub C(\den{\varphi})$. By the definition of $C$, $C(\den{\varphi})\sub \den{\varphi}$.
\end{proof}

\begin{lemma}\label{lem:diamond}
For any $X, Y\in \mc{P}(FS(T^\circ))$, the following hold:

$(1)$ if $X\sub Y$, then $\D X\sub \D Y$.

$(2)$ $\D X\sub Y$ if and only if $X\sub \bb Y$.

$(3)$ $\D C(X)\sub C(\D X)$.

$(4)$ $C(\D^3 X)\sub C(\D^2 X)$.
\end{lemma}
\begin{proof}
$(1)$ Assume $X\sub Y$. Suppose $\tup{\Gamma}\in \D X$ with $\Gamma\in X$. Then $\Gamma\in Y$. Then $\tup{\Gamma}\in \D Y$.

$(2)$ Assume $\D X\sub Y$. Suppose $\Gamma\in X$. Then $\tup{\Gamma}\in \D X$. Then $\tup{\Gamma}\in Y$. Then $\Gamma\in \bb Y$. Hence $X\sub \bb Y$.
Assume $X\sub \bb Y$. Suppose $\tup{\Gamma}\in \D X$ with $\Gamma\in X$. Then $\Gamma\in \bb Y$. Then $\tup{\Gamma}\in Y$. Hence $\D X\sub Y$.

$(3)$ By Lemma \ref{lem:C1} $(2)$, let $C(\D X)=\den{\varphi}$ with $\varphi\in T^\circ$. Then $\D X\sub \den{\varphi}$.
Let $\mf{X}=\{\tup{\chi}^n \in X\mid \chi\in T^\D~\&~n\geq 0\}$.
Suppose $\mf{X}=\ve$. Then $X\sub FS(T^\bu)\sub \den{\bot}$. 
Hence $C(X)\sub \den{\bot}$. By $(1)$, $\D C(X)\sub \D\den{\bot}$. It is easy to show that $\D\den{\bot}\sub\den{\bot}$. Clearly $\den{\bot}\sub\den{\varphi}$. Hence $\D C(X)\sub C(\D X)$. Suppose $\mf{X}\neq\ve$. Take any $\tup{\chi_i}^{n_i}\in X$ with $\chi_i\in T^\D$. Since $\D X\sub \den{\varphi}$, we have $\tup{\chi_i}^{n_i+1}\in \den{\varphi}$. Then $\msf{G}\vdash\tup{\chi_i}^{n_i+1}\seq_{T^\circ}\varphi$. By Lemma \ref{interpolant}, there exists $\gamma_i\in T^\circ$ with $\msf{G}\vdash \tup{\chi_i}^{n_i}\seq\gamma_i$ and $\msf{G}\vdash \tup{\gamma_i} \seq\varphi$.
Let $Y\sub T^\circ$ be the set of all such interpolants. Since $|T^\circ|_\msf{G}$ is finite, 
$|Y|_\msf{G}$ is finite. Let $\{\gamma_0,\ldots,\gamma_k\}$ be the set of all representatives selected from equivalence classes in $|Y|_\msf{G}$. Let $\xi=\gamma_0\vee\ldots\vee\gamma_k$. By $({\seq}\vee)$ and $(\vee{\seq})$, (i) $\msf{G}\vdash \tup{\chi_i}^{n_i}\seq\xi$ and (ii) $\msf{G}\vdash \tup{\xi} \seq\varphi$.
Then $\mf{X}\sub \den{\xi}$. Clearly $FS(T^\bu)\sub \den{\xi}$. Hence $X\sub \den{\xi}$. Then $C(X)\sub \den{\xi}$. By $(1)$, $\D C(X)\sub \D\den{\xi}$. It suffices to show that $\D\den{\xi}\sub \den{\varphi}$. Take any $\tup{\tup{\alpha}^j}\in\D\den{\xi}$ with 
$\tup{\alpha}^j\in \den{\xi}$.
If $\alpha\in T^\bu$, then $\tup{\tup{\alpha}^j}\in \den{\varphi}$. Suppose $\alpha\in T^\D$. Then (iii) $\msf{G}\vdash \tup{\alpha}^j\seq_{T^\circ}\xi$. By applying $(Cut)$ to (iii) and (ii), $\msf{G}\vdash \tup{\tup{\alpha}^j}\seq_{T^\circ}\varphi$. Then $\tup{\tup{\alpha}^j}\in \den{\varphi}$. Hence $\D\den{\xi}\sub \den{\varphi}$. Therefore $\D C(X)\sub C(\D X)$.

$(4)$ By Lemma \ref{lem:C1} $(2)$, let $C(\D^2 X)=\den{\varphi}$ for some $\varphi\in T^\circ$. By Lemma \ref{lem:C1} $(3)$, $\D^2 X\sub \den{\varphi}$. Assume $\tup{\psi}^{n+3}\in \D^3 X$ with $\tup{\psi}^n\in X$. If $\psi\in T^\bu$, then $\tup{\psi}^{n+3}\in \den{\varphi}$. Suppose $\psi\in T^\D$. Clearly $\tup{\psi}^{n+2}\in\D^2 X$. Then $\tup{\psi}^{n+2}\in\den{\varphi}$.
Hence $\msf{G}\vdash \tup{\psi}^{n+2}\seq_{T^\circ}\varphi$. By Lemma \ref{rule:R23}, $\msf{G}\vdash \tup{\psi}^{n+3}\seq_{T^\circ}\varphi$. Then $\tup{\psi}^{n+3}\in\den{\varphi}$. Then $\D^3X\sub \den{\varphi}= C(\D^2X)$.
Hence $C(\D^3X)\sub C(\D^2X)$.
\end{proof}

\begin{lemma}\label{lem:d23}
For any $X\in \mc{P}(FS(T^\circ))$,  $\D_c^3 X \sub \D_c^2 X$.
\end{lemma}
\begin{proof}
By Lemma \ref{lem:diamond} $(1)$ and $(3)$, $\D_c^3 X= C(\D C(\D C(\D X)))\sub C(C(C(\D^3 X)))$. By Lemma \ref{lem:C1} $(3)$ and $(5)$, $C(C(C(\D^3 X))) = C(\D^3X)$. By Lemma \ref{lem:diamond} $(4)$, $C(\D^3X)\sub C(\D^2X)$. By Lemma \ref{lem:C1} $(3)$, $\D X\sub C(\D X)$. By Lemma \ref{lem:diamond} $(2)$, $\D^2 X\sub \D C(\D X)$. By Lemma \ref{lem:C1} $(4)$, $C(\D^2 X)\sub C(\D C(\D X))= \D_c^2 X$. Hence $\D_c^3 X\sub \D_c^2 X$.
\end{proof}

\begin{lemma}\label{lem:adj}
For any $X, Y\in \mc{P}(FS(T^\circ))$, $\D_c C(X)\sub C(Y)$ if and only if $C(X)\sub \bb C(Y)$.
\end{lemma}
\begin{proof}
Assume $\D_c C(X)\sub C(Y)$. Then $C(\D C(X))\sub C(Y)$. 
Clearly $\D C(X)\sub \D C(C(X))$. By Lemma \ref{lem:diamond} $(3)$, $\D C(X)\sub \D C(C(X))\sub C(\D C(X))$.
Then $\D C(X)\sub C(Y)$. By Lemma \ref{lem:diamond} $(2)$, $C(X)\sub \bb C(Y)$. Assume $C(X)\sub \bb C(Y)$. By Lemma \ref{lem:diamond} $(2)$, $\D C(X)\sub C(Y)$. By Lemma \ref{lem:C1} $(4)$ and $(5)$, $C(\D C(X))\sub C(C(Y))\sub C(Y)$. Hence $\D_c C(X)\sub C(Y)$.
\end{proof}

\begin{corollary}\label{cor:23}
For any $\den{\varphi},\den{\psi}\in \den{T^\circ}$, $\D_c\den{\varphi}\leq^\circ\den{\psi}$ if and only if $\den{\varphi}\leq^\circ\bb\den{\psi}$
\end{corollary}
\begin{proof}
By Lemma \ref{lem:adj} and Lemma \ref{lem:C1} $(6)$.
\end{proof}

\begin{lemma}\label{lem:D:bb}
For any $\varphi, \D\varphi, \bb\varphi\in T^\circ$, $(1)$ $\D_c\den{\varphi}=\den{\D\varphi}$ and $(2)$ $\bb\den{\varphi}=\den{\bb\varphi}$.
\end{lemma}
\begin{proof}
$(1)$ Assume $\tup{\psi}^n\in \den{\varphi}$. Then $\tup{\psi}^{n+1}\in \D\den{\varphi}$. If $\psi\in T^\bu$, then $\tup{\psi}^{n+1}\in\den{\D\varphi}$. Suppose $\psi\in T^\D$. Then $\msf{G}\vdash \tup{\psi}^n\seq_{T^\circ}\varphi$. By $({\seq}\D)$, $\msf{G}\vdash \tup{\psi}^{n+1}\seq_{T^\circ}\D\varphi$. Then $\tup{\psi}^{n+1}\in\den{\D\varphi}$. Then $\D\den{\varphi}\sub \den{\D\varphi}$. Hence $C(\D\den{\varphi})\sub \den{\D\varphi}$, i.e., $\D_c\den{\varphi}\sub\den{\D\varphi}$. Assume $\tup{\chi}^n\in\den{\D\varphi}$.
By Lemma \ref{lem:C1} $(2)$, let $\D_c\den{\varphi}=C(\D\den{\varphi})=\den{\delta}$ for some formula $\delta\in T^\circ$.
If $\chi\in T^\bu$, then $\tup{\chi}^n\in \den{\delta}$.
Suppose $\chi\in T^\D$. Then (i) $\msf{G}\vdash\tup{\chi}^n\seq_{T^\circ}\D\varphi$. Clearly $\varphi\in\den{\varphi}$
and $\varphi,\D\varphi\in T^\D$. Then $\tup{\varphi}\in \D\den{\varphi}$. By Lemma \ref{lem:C1} $(3)$, $\D\den{\varphi}\sub C(\D\den{\varphi})=\den{\delta}$. Then $\tup{\varphi}\in \den{\delta}$. Then $\msf{G}\vdash \tup{\varphi}\seq_{T^\circ}\delta$. By $(\D{\seq})$, (ii) $\msf{G}\vdash \D\varphi\seq_{T^\circ}\delta$. By applying $(Cut)$ to (i) and (ii), $\msf{G}\vdash\tup{\chi}^n\seq_{T^\circ}\delta$. Then $\tup{\chi}^n\in\den{\delta}$. Hence $\den{\D\varphi}\sub \den{\delta} = \D_c\den{\varphi}$.

$(2)$ Assume $\tup{\psi}^n\in\bb\den{\varphi}$. Then $\tup{\psi}^{n+1}\in\den{\varphi}$. If $\psi\in T^\bu$, then $\tup{\psi}^{n}\in\den{\bb\varphi}$. Suppose $\psi\in T^\D$. Then $\msf{G}\vdash \tup{\psi}^{n+1}\seq_{T^\circ}\varphi$. By $({\seq}\bb)$, $\msf{G}\vdash \tup{\psi}^{n}\seq_{T^\circ}\bb\varphi$. Then $\tup{\psi}^{n}\in\den{\bb\varphi}$. Hence $\bb\den{\varphi}\sub\den{\bb\varphi}$.
Assume $\tup{\chi}^{n}\in \den{\bb\varphi}$. If $\chi\in T^\bu$, then $\tup{\chi}^{n+1}\in \den{\varphi}$. Suppose $\chi\in T^\D$. Then $\msf{G}\vdash \tup{\chi}^n \seq_{T^\circ} \bb\varphi$. Clearly $\msf{G}\vdash \tup{\bb\varphi} \seq_{T^\circ} \varphi$. By $(Cut)$, $\msf{G}\vdash \tup{\chi}^{n+1} \seq_{T^\circ} \varphi$. Then $\tup{\chi}^{n+1}\in \den{\varphi}$. Then $\tup{\chi}^{n}\in \bb\den{\varphi}$. Hence $\den{\bb\varphi}\sub\bb\den{\varphi}$.
\end{proof}

\begin{lemma}\label{lem:box}
For any $\psi\in T^\circ$, there exists a formula $\varphi\in T^\circ$ with $\bb \den{\psi}=\den{\varphi}$.
\end{lemma}
\begin{proof}
Let $\mf{X}=\{\tup{\chi}^n \in \bb\den{\psi}\mid \chi\in T^\D~\&~n\geq 0\}$. Clearly $\bot\in T^\D$ and $\tup{\bot}\in \den{\psi}$. Then $\bot\in \bb\den{\psi}$. Hence $\mf{X}\neq\ve$. Take any $\tup{\chi_i}^{n_i}\in\mf{X}$. Then $\tup{\chi_i}^{n_i+1}\in \den{\psi}$. Hence $\msf{G}\vdash\tup{\chi_i}^{n_i+1}\seq_{T^\circ}\psi$. By Lemma \ref{interpolant}, there is $\gamma_i\in T^\circ$ with $\msf{G}\vdash \tup{\chi_i}^{n_i}\seq\gamma_i$ and $\msf{G}\vdash \tup{\gamma_i} \seq\psi$.
Let $Y\sub T^\circ$ be the set of all such interpolants. Since $|T^\circ|_\msf{G}$ is finite, $|Y|_\msf{G}$ is finite. Let $\{\gamma_0,\ldots,\gamma_k\}$ be the set of all representatives selected from equivalence classes in $|Y|_\msf{G}$. Let $\xi=\gamma_0\vee\ldots\vee\gamma_k$. 
By $({\seq}\vee)$ and $(\vee{\seq})$, (i) $\msf{G}\vdash \tup{\chi_i}^{n_i}\seq\xi$ and (ii) $\msf{G}\vdash \tup{\xi} \seq\psi$. Then $\mf{X}\sub \den{\xi}$. Clearly $FS(T^\bu)\sub \den{\xi}$.
Hence $\bb\den{\psi}\sub \den{\xi}$.
Take any $\tup{\alpha}^k\in \den{\xi}$. Then $\tup{\alpha}^{k+1}\in \D\den{\xi}$. Suppose $\alpha\in T^\bu$. Then $\tup{\alpha}^{k+1}\in \den{\psi}$. Hence $\tup{\alpha}^{k}\in \bb\den{\psi}$.
Suppose $\alpha\in T^\D$. Then (iii) $\msf{G}\vdash \tup{\alpha}^k\seq_{T^\circ}\xi$. By applying $(Cut)$ to (ii) and (iii), $\msf{G}\vdash \tup{\alpha}^{k+1}\seq_{T^\circ}\psi$. Then $\tup{\alpha}^{k+1}\in\den{\psi}$. Then $\tup{\alpha}^{k}\in \bb\den{\psi}$. Hence $\den{\xi}\sub \bb\den{\psi}$. Therefore $\bb\den{\psi}= \den{\xi}$.
\end{proof}

By Lemma \ref{lem:C1} $(2)$ and Lemma \ref{lem:box}, the operations $\D_c$ and $\bb$ are unary operations on $\den{T^\circ}$. Now we get the algebra $\mf{A}_T = (\den{T^\circ}, \wedge^\circ, \vee^\circ, \neg^\circ, \bot^\circ, \top^\circ, \D_c, \bb)$.

\begin{lemma}\label{lem:finite}
$\mf{A}_T$ is a finite quasi-transitive tense algebra.
\end{lemma}
\begin{proof}
By Lemma \ref{lem:boole}, the $(\wedge^\circ, \vee^\circ, \neg^\circ, \bot^\circ, \top^\circ)$-reduct of $\mf{A}_T$ is a finite Boolean algebra. 
By Corollary \ref{cor:23}, $\mf{A}_T$ is a tense algebra. By Lemma \ref{lem:d23}, $\D_c^3\den{\varphi}\leq^\circ \D_c^2\den{\varphi}$
for any $\den{\varphi}\in \den{T^\circ}$. Then
$\mf{A}_T$ is quasi-transitive.
\end{proof}

\begin{theorem}\label{thm:G:fmp}
$\msf{G}$ has the FAP.
\end{theorem}
\begin{proof}
Assume $\msf{G}\not\vdash\tup{\varphi}^n\seq\psi$.
Clearly $\msf{G}\vdash\tup{\varphi}^n\seq\D^n\varphi$.
Then $\msf{G}\not\vdash\D^n\varphi\seq\psi$.
Let $\chi=\D^n\varphi$ and $T = Sub(\chi)\cup Sub(\psi)\cup\{\top,\bot\}$. Then $\msf{G}\not\vdash\chi\seq_{T^\circ}\psi$.
Let $\theta_T$ be the assignment in $\mf{A}_T$ with $\theta_T(p)=\den{p}$ for every propositional variable $p$.
By the definition of $\mf{A}_T$ and Lemma \ref{lem:D:bb}, one can easily show by induction on the complexity of a formula $\xi\in T^\circ$ that $\w{\theta_T}(\xi)=\den{\xi}$. Since $\chi\in T$, we have $\chi\in T^\D$. By $\msf{G}\vdash\chi\seq_{T^\circ}\chi$, we have $\chi\in\den{\chi}$. By $\msf{G}\not\vdash\chi\seq_{T^\circ}\psi$, we have $\chi\not\in \den{\psi}$. Then $\den{\chi}\not\sub\den{\psi}$.
Hence $\mf{A}_T\not\models\chi\seq\psi$, i.e., $\mf{A}_T\not\models\tup{\varphi}^n\seq\psi$. By Lemma \ref{lem:finite},
$\mf{A}_T$ is a finite quasi-transitive tense algebra.
\end{proof}

Finally, by the algebraic completeness of tense logic $\mathsf{Kt}_2^3=\msf{Kt}\oplus \D^3 p\imp \D^2 p$ (cf.~\cite{BDV2001,Yde2007}),  for any tense formula $\varphi$, $\msf{Kt}_2^3\vdash\varphi$ if and only if $\msf{G}\vdash\top\seq\varphi$. Using duality between tense algebras and bidirectional frames (cf.~\cite{BDV2001,Yde2007}), by Theorem \ref{thm:G:fmp}, one can get the FMP of $\mathsf{Kt}_2^3$, i.e., if $\msf{Kt}_2^3\not\vdash\varphi$, then $\varphi$ is refuted by the dual frame of the algebra $\mf{A}_T$ where $T$ is the set of tense formulas $Sub(\varphi)\cup\{\top,\bot\}$. Finally, since $\mathsf{Kt}_2^3$ is a conservative extension of $\mathsf{K}_2^3=\msf{K}\oplus \D^3 p\imp \D^2 p$, we obtain the FMP of $\msf{K}_2^3$.

\begin{corollary}\label{thm:fmp}
$\mathsf{Kt}_2^3$ and $\mathsf{K}_2^3$ have the FMP and hence are decidable.
\end{corollary}

\section{Concluding remarks}
We established the finite model property of the quasi-transitive modal logic $\msf{K}\oplus \D\D\D p\imp\D\D p$ by showing the finite model property of its conservative tense extension $\mathsf{Kt}_2^3$. In the sequent calculus $\msf{G}$ for the tense logic $\mathsf{Kt}_2^3$, for each sequent which is not derivable in $\msf{G}$ there exists a finite syntactic algebraic model that refutes the sequent. We can extend the method in the present work to show the FMP of logics $\msf{K}\oplus \B^n p\imp \B^m p$ for $n\neq m\geq 0$. Furthermore, we may extend the method to show the FMP of non-classical modal logics. For example, the finite model property of some intuitionistic modal logics and lattice-based modal logics can be proved.

\section*{Acknowledgements}
The first author was supported by Chinese National Funding of Social Sciences (No. 17CZX048). The second author was
supported by Guangdong Province Higher Vocational Colleges $\&$ Schools Pearl River Scholar Funded Scheme (2017-2019).
Thanks are given to the reviewers' insightful and helpful comments on the revision of this paper. In particular, the first reviewer mentioned some proof-theoretic points of the sequent calculus in the first version. The second reviewer pointed out the possibility of a shorter proof of the main result.

\bibliographystyle{plain}

\end{document}